\documentclass[conference, 10pt]{IEEEtran} 
\usepackage{graphicx}
\usepackage{epsfig}
\usepackage{epstopdf}
\usepackage{amsmath,amsthm,amsfonts,amssymb}
\usepackage{subcaption}
\usepackage{multirow}
\usepackage{cite}
\usepackage{tikz}
\usepackage{ellipsis}
\usepackage{xcolor}
\usepackage{xifthen}
\usepackage{cancel}
\usepackage{amsmath,bm,bbm}
\usepackage{balance}
\usepackage{array}
\usepackage[linesnumbered, ruled, vlined, noend]{algorithm2e}

\newcommand{\realnumber}{\mathbb{R}}
\theoremstyle{definition} \newtheorem{defi}{Definition}
\theoremstyle{definition} 
\theoremstyle{definition} 
\theoremstyle{definition} \newtheorem{lemma}{Lemma}
\newtheorem{theorem}{Theorem}

\newcommand{\Aut}{\mathit{Aut}}

\newcommand{\GA}{\mathsf{GA}}

\newcommand{\GL}{\mathsf{GL}}
\newcommand{\LTA}{\mathsf{LTA}}
\newcommand{\BLTA}{\mathsf{BLTA}}

\newcommand{\EC}{\mathsf{EC}}

\newcommand{\UTL}{\mathsf{UTL}}

\newcommand{\PL}{\mathsf{PL}}

\DeclareMathOperator{\scdec}{SC}
\DeclareMathOperator{\ascdec}{aSC}

\DeclareMathOperator{\dec}{dec}
\DeclareMathOperator{\adec}{adec}
\newcommand{\pid}{\mathbbm{1}}
\newcommand{\fixme}[2]{\ifx&#2&{\leavevmode\color{red}#1}\else{\leavevmode\color{red}FIXME\{}#1{\leavevmode\color{red}\}}\footnote{{\leavevmode\color{red}#2}}\PackageWarning{Fixme}{#1: #2}\fi}
\newcolumntype{M}[1]{>{\centering\arraybackslash}m{#1}}

\newcommand*{\PROOFS}{}%

\begin{document}

\title{Classification of Automorphisms \\ for the Decoding of Polar Codes}

\author{\IEEEauthorblockN{Charles Pillet, Valerio Bioglio, Ingmar Land}
\IEEEauthorblockA{Mathematical and Algorithmic Sciences Lab\\ Paris Research Center, Huawei Technologies France SASU \\
Email: $\{$charles.pillet1,valerio.bioglio,ingmar.land$\}$@huawei.com}} 

\maketitle

\begin{abstract}
This paper proposes new polar code design principles for the low-latency automorphism ensemble (AE) decoding. 
Our proposal permits to design a polar code with the desired automorphism group (if possible) while assuring the decreasing monomial property. 
Moreover, we prove that some automorphisms are redundant under AE decoding, and we propose a new automorphisms classification based on equivalence classes.
Finally, we propose an automorphism selection heuristic based on drawing only one element of each class; we show that this method enhances the block error rate (BLER) performance of short polar codes even with a limited number of automorphisms.
\end{abstract}


\section{Introduction}\label{sec:intro}

Polar codes \cite{ArikanFirst} are a class of linear block codes relying on the phenomenon of channel polarization.  
They are shown to be capacity-achieving on binary memoryless symmetric channels under successive cancellation (SC) decoding for infinite block length.  
However, in the finite-length regime, SC decoding is far from maximum likelihood (ML) decoding.  
SC list (SCL) decoding has been proposed in \cite{TalSCL} to overcome this problem. 
Since the correct candidate codeword may not be chosen due to the minimizing metric, a cyclic redundancy check (CRC) code is concatenated to the code, acting as a genie that selects the correct codeword regardless the metric. 
This scheme, referred as CRC-aided SCL (CA-SCL), is now the state-of-the-art decoding algorithm for polar code. 

The main drawback of SCL algorithm is its decoding latency due to the information exchange performed at each bit-decoding step among parallel SC decoders. 
In order to avoid this decoding delay, permutation-based decoding for SC was proposed in \cite{PermGross}; a similar approach was proposed in \cite{BPLRM} for belief propagation (BP) and in \cite{SCANL} for soft cancellation (SCAN). 
According to this framework, $M$ instances of the same decoder are ran in parallel on permuted factor graphs of the code; however, the performance gain is poor since each factor graph permutation (FGP) alters the bits polarization, and hence the frozen set of the code. 
As a consequence, research towards permutations not altering the frozen set were carried out \cite{PermDecRussian}; such permutations are called \emph{automorphisms} and form the group of permutations mapping a codeword into another codeword.
The automorphism group of binary Reed-Muller code is known to be the general affine group \cite{WilliamsSloane} and automorphisms were used as permutations with BP as inner decoder in \cite{geiselhart2020automorphism}. 
Due to their close relation with polar codes SC was also used; this new decoding approach is referred to as \emph{automorphism ensemble} (AE) decoding.

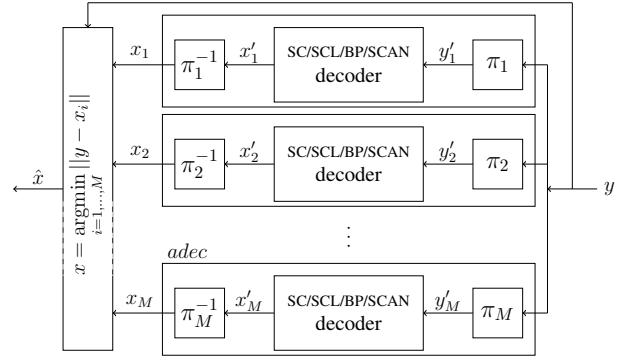
\begin{figure}[t!]
	\centering
	\resizebox{0.45\textwidth}{!}{\begin{tikzpicture}

\draw[ ] (8.75,0) node [font=\large] {$y$};
\draw[<-]  (7.5,0)  -- (8.5,0) ;
\draw  (-1.75,3.75)  -- (8,3.75) ;
\draw  (8,0)  -- (8,3.75) ;
\draw[->]  (-1.75,3.75)  -- (-1.75,3.25) ;
\draw  (7.5,-2.5)  -- (7.5,2.5) ;
\draw[<-]  (7,2.5)  -- (7.5,2.5) ;
\draw[<-]  (7,0.5)  -- (7.5,0.5) ;
\draw[<-]  (7,-2.5)  -- (7.5,-2.5) ;
\draw[<-]  (5,2.5)  -- (6,2.5) ;
\draw[ ] (5.5,2.75) node [font=\large] {$y'_1$};
\draw[<-]  (5,0.5)  -- (6,0.5) ;
\draw[ ] (5.5,0.75) node [font=\large] {$y'_2$};
\draw[<-]  (5,-2.5)  -- (6,-2.5) ;
\draw[ ] (5.5,-2.25) node [font=\large] {$y'_M$};
\draw[<-]  (1,2.5)  -- (2,2.5) ;
\draw[ ] (1.5,2.75) node [font=\large] {$x'_1$};
\draw[<-]  (1,0.5)  -- (2,0.5) ;
\draw[ ] (1.5,0.75) node [font=\large] {$x'_2$};
\draw[<-]  (1,-2.5)  -- (2,-2.5) ;
\draw[ ] (1.5,-2.25) node [font=\large] {$x'_M$};
\draw[<-]  (-1.25,2.5)  -- (0,2.5) ;
\draw[ ] (-0.7,2.75) node [font=\large] {$x_1$};
\draw[<-]  (-1.25,0.5)  -- (0,0.5) ;
\draw[ ] (-0.7,0.75) node [font=\large] {$x_2$};
\draw[<-]  (-1.25,-2.5)  -- (0,-2.5) ;
\draw[ ] (-0.7,-2.25) node [font=\large] {$x_M$};
\draw[<-]  (-3.25,-0)  -- (-2.25,-0) ;
\draw[ ] (-2.75,0.25) node [font=\large] {$\hat{x}$};

\draw[fill=white]  (0,3) rectangle (1,2);
\draw[ ] (0.5,2.5) node [font=\Large] {$\pi_1^{-1}$};
\draw[fill=white]  (0,1) rectangle (1,0);
\draw[ ] (0.5,0.5) node [font=\Large] {$\pi_2^{-1}$};
\draw[fill=white]  (0,-2) rectangle (1,-3);
\draw[ ] (0.5,-2.5) node [font=\Large] {$\pi_M^{-1}$};

\draw[fill=white]  (2,3.25) rectangle (5,1.75);
\draw[ ] (3.5,2.5) node [font=\large] {$\begin{array}{c} \text{\small SC/SCL/BP/SCAN} \\ \text{decoder} \end{array}$};
\draw[fill=white]  (2,1.25) rectangle (5,-0.25);
\draw[ ] (3.5,0.5) node [font=\large] {$\begin{array}{c} \text{\small SC/SCL/BP/SCAN} \\ \text{decoder} \end{array}$};
\path (3.5,1) -- (3.5,-3) node [font=\large, midway, sloped] {$\dots$};
\draw[fill=white]  (2,-1.75) rectangle (5,-3.25);
\draw[ ] (3.5,-2.5) node [font=\large] {$\begin{array}{c} \text{\small SC/SCL/BP/SCAN} \\ \text{decoder} \end{array}$};

\draw[fill=white]  (6,3) rectangle (7,2);
\draw[ ] (6.5,2.5) node [font=\Large] {$\pi_1$};
\draw[fill=white]  (6,1) rectangle (7,0);
\draw[ ] (6.5,0.5) node [font=\Large] {$\pi_2$};
\draw[fill=white]  (6,-2) rectangle (7,-3);
\draw[ ] (6.5,-2.5) node [font=\Large] {$\pi_M$};

%

\draw[fill=white]  (-2.25,3.25) rectangle (-1.25,-3.25);
\draw[dashed, white]  (-2.25,-0.25) edge (-2.25,-1.75);
\draw[dashed, white]  (-1.25,-0.25) edge (-1.25,-1.75);
\draw[ ] (-1.75,0.1) node [rotate=90,font=\large] {$x=\underset{i=1,\ldots, M}{\mathrm{argmin}}\,||y-x_i||$};

\draw  (-0.25,3.5) rectangle (7.25,1.65);
\draw  (-0.25,1.5) rectangle (7.25,-.35);
\draw  (-0.25,-1.5) rectangle (7.25,-3.35);
\draw[ ] (0.25,-1.25) node [font=\large] {$adec$};

\end{tikzpicture}}
	\caption{Structure of the automorphism ensemble (AE) decoder.}
	\label{fig:AE}
\end{figure} 

Lower-triangular affine (LTA) transformations form a subgroup of the automorphism group of polar codes \cite{BardetPolyPC}. 
However, these transformations commute with SC decoding \cite{geiselhart2020automorphism}, leading to no gain under AE decoding.
In \cite{geiselhart2021automorphismPC,li2021complete}, the block-lower-triangular affine (BLTA) group was proved to be the complete affine automorphism group of a polar code. 
In \cite{geiselhart2021automorphismPC,PC_UTL_design}, BLTA transformations were successfully applied to AE decoding of polar codes.
Finally, the author in \cite{Perm_Russian_polar_subcode} uses polar subcodes in conjunction with a costly choice of permutation set to design good polar codes for AE decoding. 

In this paper, we propose a method to design polar codes having a desired affine automorphism group.  
We prove that LTA is not always the complete SC absorption group, namely that a larger set of automorphisms may be absorbed under SC decoding. 
Moreover, we introduce the concept of \emph{redundant} automorphisms, providing the maximum number of permutations providing possibly different codeword candidates under AE-SC. 
Finally, we present a small automorphism set design for AE-SC decoding avoiding redundancy and exhibiting good performance compared to ML and CA-SCL decoding.
\ifdefined\PROOFS
	This extended version of the paper includes the proofs of the lemmas that are omitted in the conference version of the paper due to space limitations.
\else
	The proofs of the lemmas are based on the application of group theory properties and are omitted for space limitation; they will be included in an extended version of the paper. 
\fi

\section{Preliminaries}\label{sec:pre}
\subsection{Polar codes}
An $(N,K)$ polar code of length $N=2^n$ and dimension $K$ is a binary block code defined by the kernel matrix $T_2\triangleq [ \begin{smallmatrix} 1 & 0\\ 1 & 1 \end{smallmatrix} ]$, the transformation matrix $T_N = T_2^{\otimes n}$, an information set $\mathcal{I} \in [N]$ and a frozen set $\mathcal{F} = [N] \backslash \mathcal{I}$, $[N] = \{0,1,\ldots,N-1\}$.  
For encoding, an input vector $u = ( u_0,u_1,\dots, u_{N-1} )$ is generated by assigning $u_i = 0$ for $i \in \mathcal{F}$ (frozen bits), and storing information in the remaining entries. 
The codeword is then computed as $x = u \cdot T_N$.
The information set is commonly selected according to the reliabilities of the virtual bit-channels resulting from the polarization, which can be determined through different methods \cite{frozenset}. 

SC decoding is the fundamental decoding algorithm for polar codes and is proved to be capacity-achieving at infinite block length \cite{ArikanFirst}. 
BP decoding is a popular message passing decoder conceived for codes defined on graphs, and can be easily adapted to polar codes. 
SCAN \cite{SCANL} is an iterative SISO decoder using the SC schedule with the BP update rules. 
The three decoders were used in permutation decoding\cite{PermGross,BPLRM,SCANL} to enhance performance without increasing latency.

\subsection{Monomial codes}\label{subsec:monomial}
Monomial codes of length $N=2^n$ are a family of codes that can be obtained as evaluations of monomials in $n$ binary variables. 
Polar and Reed-Muller codes can be described through this formalism \cite{WilliamsSloane}; in fact, the rows of $T_N$ represent all possible evaluations of monomials over $\mathbb{F}_2^n$ \cite{BardetPolyPC}.
A monomial code of length $N$ and dimension $K$ is generated by $K$ monomials out of the $N$ monomials over $\mathbb{F}_2^n$. 
These $K$ chosen monomials form the \emph{generating monomial set} $\mathcal{G}$ of the code, while the linear combinations of their evaluations provide the codebook of the code. 
A monomial code is called \emph{decreasing} if $\mathcal{G}$ includes all factors and antecedents of every monomial in the set \cite{BardetPolyPC}. 
In this case, $\mathcal{G}$ can be also described with the minimal information set $\mathcal{I}_{min}$ containing the few monomials necessary to retrieve all the others. 
This structure is equivalent to the partially-symmetric monomial code construction \cite{sym1}. 

Reed-Muller codes are monomial codes generated by all monomials up to a certain degree. 
Polar codes select generating monomials following the polarization effect; if the polar code design is compliant with the universal partial order (UPO) framework, the resulting code is provably decreasing monomial \cite{BardetPolyPC}. 
In the following, we assume that the polar code is a decreasing monomial code. 

\subsection{Affine automorphism subgroups}\label{subsec:AUT_GA}
An automorphism $\pi$ of a code $\mathcal{C}$ is a permutation of $N$ elements mapping every codeword $x \in \mathcal{C}$ into another codeword $\pi(x) \in \mathcal{C}$. 
The automorphism group $\Aut(\mathcal{C})$ of a code $\mathcal{C}$ is the set containing all automorphisms of the code. 
For monomial codes, the \emph{affine automorphism group} $\mathcal{A} \subseteq \Aut(\mathcal{C})$, formed by the automorphism that can be written as affine transformations of $n$ variables, is of particular interest. 
An affine transformation of $n$ variables is described by equation 
\begin{align}
	z & \mapsto z' = A z + b ,
	\label{eq:GA}
\end{align}
$z , z' \in \mathbb{F}_2^n$, where the transformation matrix $A$ is an $n\times n$ binary invertible matrix and $b$ is a binary column vector of length $n$. 
The variables in \eqref{eq:GA} are the binary representations of code bit indices, and thus affine transformations represent code bit permutations.
The automorphism group of Reed-Muller codes is known to be the complete affine group GA$(n)$ \cite{WilliamsSloane}, while the affine automorphism group of polar codes has been recently proved to be the block-lower-triangular affine (BLTA) group \cite{geiselhart2021automorphismPC,li2021complete}, namely the group of affine transformations having BLT transformation matrix. 
The BLTA group is defined by the block structure $S=(s_1,\ldots,s_t)$, $s_1+\ldots+s_t=n$ of the admissible BLT transformation matrix, representing the sizes of the blocks alongside the diagonal. 
The last block size $s_t$ represents the symmetry of the code \cite{sym1}; this value is $n$ for Reed-Muller codes and usually 1 for polar codes.
Finally, we denote by $\mathcal{L}$, $\mathcal{U}$ and $\mathcal{P}$, the sets of lower-triangular, upper-triangular and permutation matrices of $n$ elements. 
If $A$ in \eqref{eq:GA} is replaced by $L \in \mathcal{L}$ or $U \in \mathcal{U}$, the transformation is called lower-triangular-affine (LTA) and upper-triangular-affine (UTA) respectively. 
If $b$ is removed, the transformation is \emph{linear} and corresponds to LTL, UTL and PL transformations with respectively $L$, $U$, and $P\in\mathcal{P}$ replacing $A$. 

\subsection{Automorphism Ensemble (AE) decoder}\label{subsec:AEdecoder}
An \emph{automorphism decoder} is a decoder run on a received signal that is scrambled according to a code automorphism; the result is then scrambled back to retrieve the original codeword estimation. 
More formally, given a decoder $\dec$ for a code $\mathcal{C}$, the corresponding automorphism decoder $\adec$ is given by
\begin{equation}
\label{eq:adec}
\adec(y,\pi) = \pi^{-1}\left( \dec(\pi(y)) \right),
\end{equation}
where $y$ is the received signal and $\pi \in \Aut(\mathcal{C})$. 
An \emph{automorphism ensemble} (AE) decoder, originally proposed in \cite{geiselhart2020automorphism} for Reed-Muller codes, consists of $M$ automorphism decoders running in parallel, as depicted in Figure~\ref{fig:AE}, where the codeword candidate is selected using a least-squares metric. 

AE decoding with SC component decoders can be used for polar codes, however with particular attention on the choice of the automorphisms: 
in fact, $\LTA$ automorphisms are \emph{absorbed} by SC decoding, namely automorphism SC (aSC) decoding where $\pi$ is $\LTA$ provides the same result as plain SC decoding \cite{geiselhart2020automorphism}, i.e. $\ascdec(y,\pi) = \scdec(y)$.
However, BLTA automorphisms are generally not absorbed by SC decoding \cite{geiselhart2020automorphism}, hence they can be successfully used in AE-SC decoders instead \cite{geiselhart2021automorphismPC,PC_UTL_design}. 



\section{Equivalence classes of automorphisms}\label{sec:EC_AE}
In this section, we show how to classify automorphisms into equivalent classes (EC) containing permutations providing the same results under AE-SC decoding. 
The number of ECs, corresponding to the maximum number of non-redundant automorphisms, will be investigated given a certain $\BLTA$ block structure $S$. 
Finally, we use this classification to provide an automorphism set design ensuring no redundancy.


\subsection{Decoder equivalence}\label{subsec:EC}
To begin with, we introduce the notion of decoder equivalence. 
This concept is used to cluster the automorphisms in sets always providing the same results under AE decoding. 
\begin{defi}[Decoder equivalence]
\label{def:dec-equiv}
Two automorphisms $\pi_1,\pi_2 \in \mathcal{A}$ are equivalent with respect to a decoding algorithm $\dec$, written as $\pi_1 \sim \pi_2$, if for all $y \in \realnumber^N$ 
\begin{equation}
	\adec(y;\pi_1) = \adec(y;\pi_2). 
\end{equation}
\end{defi}	
This is an equivalence relation, since it is reflexive, symmetric and transitive. 
The equivalence classes are defined as
\begin{equation}
	[\pi] \triangleq \{ \pi' \in \mathcal{A} : \pi \sim \pi' \}.
\end{equation}
The equivalence class $[\mathbbm{1}]$ of the identity automorphism $\mathbbm{1}$ corresponds to the set of automorphisms absorbed by $\dec$. 
\begin{lemma}\label{lem:inv}
If $\pi \in [\pid]$, then $\pi^{-1} \in [\pid]$.
\ifdefined\PROOFS
	\begin{proof}
	Assume $\pi \in [\pid]$. Then from the definition of the EC, we have that
	\begin{equation*}
	    \dec(y) = \pi^{-1}( \dec( \pi( y ) ) )  
	    \quad\Leftrightarrow\quad  
	    \pi(\dec(y)) = \dec( \pi( y ) ) 
	\end{equation*}
	(with the latter denoting an equivariance); and so
	\begin{equation*}
		\resizebox{.99\columnwidth}{!}{
	    $\adec(y;\pi^{-1}) = \pi( \dec( \pi^{-1}( y ) ) ) = \dec( \pi( \pi^{-1} (y) ) ) = \dec(y).$
	    }
	\end{equation*}
	\end{proof}
\fi
\end{lemma}
\begin{lemma}\label{lem:absorbed-group}
The equivalence class $[\pid]$ (set of decoder-absorbed automorphisms) is a subgroup of $\mathcal{A}$, i.e., $[\pid] \leq \mathcal{A}$.
\ifdefined\PROOFS
	\begin{proof}
	We prove this lemma using the subgroup test, stating that $[\pid] \leq \mathcal{A}$ if and only if $\forall \pi,\sigma \in [\pid]$ then $\pi^{-1}\sigma \in [\pid]$:
	\begin{align}
	    \adec(y;\pi^{-1}\sigma) &= (\pi^{-1}\sigma)^{-1}( \dec( (\pi^{-1}\sigma)( y ) ) ) =\\
	    &= \sigma^{-1}\left(\pi\left( \dec\left( \pi^{-1}\left(\sigma\left( y \right)\right)\right)\right)\right) =\\
	    &= \sigma^{-1}\left( \dec\left(\sigma\left( y \right)\right)\right) =\\
	    &= \dec\left( y \right).
	\end{align}
	\end{proof}
\fi
\end{lemma}
According to our notation, two automorphisms in the same EC always provide the same candidate under $\adec$ decoding. 
The number of non-redundant automorphisms for AE-$\dec$, namely the maximum number of different candidates listed by an AE-$\dec$ decoder, is then given by the number of equivalent classes of our relation.
\begin{lemma}\label{lem:nb_ec}
There are $|\EC|=\frac{|\mathcal{A}|}{|[\pid]|}$ equivalence classes $[\pi]$, $\pi \in \mathcal{A}$, all having the same size.
\ifdefined\PROOFS
	\begin{proof}
	Right cosets of $[\pid]$ are defined as
	\begin{equation}
		[\pid]\sigma \triangleq \{ \pi \circ \sigma : \pi \in [\pid] \}.
	\end{equation}
	Hence, permutations $\sigma_1,\sigma_2 \in [\pid]\sigma$ if and only if there exist two permutations $\pi_1,\pi_2 \in [\pid]$ such that $\sigma_1 = \pi_1 \circ \sigma$ and $\sigma_2 = \pi_2 \circ \sigma$, where the second implies $\sigma = \pi_2^{-1} \circ \sigma_2$. 
	Then 
	\begin{align*}
	\adec(y;\sigma_1) &= \adec(y;\pi_1\sigma) = \\
	&= \adec(y;\pi_1\pi_2^{-1}\sigma_2) = \\
	&= \adec(y;\sigma_2) ,
	\end{align*}
	as  $\pi_1 \pi_2^{-1} \in [\pid]$ by Lemma~\ref{lem:absorbed-group}.
	The proof is concluded by applying Lagrange’s Theorem. 
	\end{proof}
\fi
\end{lemma}
It is worth noticing that for some $y$ two ECs may produce the same candidate; however, our relation permits to calculate the maximum number of different results under $\adec$, providing an upper bound of parameter $M$ of an AE-$\dec$ decoder.

\subsection{Equivalence classes of the SC decoder}\label{subsec:EC_AESC}
In \cite{geiselhart2020automorphism}, it was shown that the group $\LTA$ is SC-absorbed i.e. $\LTA\leq[\pid]$. 
Here, we provide a larger subset for $[\pid]$.
\begin{theorem}\label{theo:EC}
	If $\BLTA(S)$ is the affine automorphism group of a polar code with $s_1>1$, then $\BLTA(2,1,\dots,1) \leq [\pid]$.
\end{theorem}
\begin{proof}
Given that $\LTA\subset\BLTA(2,1,\dots,1)$, we need to understand what happens to an automorphism $L\in\LTA$ if its entry in the first row and second column is set to 1. 
This modification $L'$ represents an extra scrambling of the first 4 entries of the codeword, which is repeated identically for every subsequent block of 4 entries of the vector. 
In practice, the difference between the SC decoding of two codewords permuted according to $L$ and $L'$ is that LLRs of the leftmost $4\times 4$ decoding block are identical but permuted. 
Then, in order to estimate the candidate codeword calculated by each SC decoder, we need to track the decoding of the $4\times 4$ block when the input LLRs are permuted. 
If the polar code follows the UPO, then each block of 4 entries of the input vector can be described as one of five sequences of frozen (F) and information (I) bits, listed in increasing rate order:
\begin{itemize}
	\item \textbf{[FFFF]}: this represents a rate-zero node, and returns a string of four zeroes no matter the input LLRs; this is independent of the permutation. 
	\item \textbf{[FFFI]}: this represents a repetition node, and returns a string of four identical bits given by the sign of the sum of the LLRs; this is independent of the permutation. 
	\item \textbf{[FFII]}: this case is not possible if $s_1>1$. 
	\item \textbf{[FIII]}: this represents a single parity check node, and returns the bit representing the sign of each LLR while the smallest LLR may be flipped if the resulting vector has even Hamming weight; permuting them back gives the same result for the two decoders. 
	\item \textbf{[IIII]}: this represents a rate-one node, and returns the bit representing the sign of each LLR; permuting them back gives the same result for the two decoders. 
\end{itemize}
As a consequence, applying such transformation matrix does not change the result of the $4\times4$ block, and thus, using the same reasoning as in \cite{geiselhart2020automorphism}, of the SC decoder.
\end{proof}
This proves that particular frozen patterns are invariant under SC decoding, and their prevalence suggests the presence of a larger absorption group. 
As an example, the $(32,23)$ polar code defined by $\mathcal{I}_{min}=\{7,9\}$ has the $\BLTA(3,2)$ affine automorphism group and $[\pid]=\BLTA(3,1,1)$; we conjecture that $[\pid]$ is always a $\BLTA$ for automorphism SC (aSC) decoding. 
However, since these cases are quite rare and do not seem to provide good polar codes, we focus our discussion on the case $[\pid]=\BLTA(2,1,\ldots,1)$. 

\begin{lemma}\label{lem:sizeBLTA}
The number of permutations of $\BLTA(S)$, with $S=(s_1,\dots,s_t)$ and $\sum_{i=1}^t s_i = n$, is: 
\begin{align}\label{eq:sizeBLTA}
	|\BLTA(S)| &= |\mathcal{L}|\cdot 2^n \cdot \prod_{i=1}^t \left( \prod_{j=2}^{s_i} \left( 2^j-1 \right) \right) \\
	 &= 2^{\frac{n(n+1)}{2}}\cdot \prod_{i=1}^t \left( \prod_{j=2}^{s_i} \left( 2^j-1 \right) \right).
\end{align}
\ifdefined\PROOFS
	\begin{proof}
	It is well known that $|\GL(m)| = \prod_{i=0}^{m-1} \left( 2^m-2^i \right)$. 
From this, we have that
	\begin{align*}
	|\GL(m)| & = \prod_{i=0}^{m-1} \left( 2^m-2^i \right) = \\
	& = \prod_{i=0}^{m-1} 2^i \left( 2^{m-i}-1 \right) = \\
	& = \left( \prod_{j=0}^{m-1} 2^j \right) \cdot \left( \prod_{i=0}^{m-1} \left( 2^{m-i}-1 \right) \right) = \\
	& = 2^{\sum_{j=0}^{m-1} j} \cdot \prod_{i=1}^{m} \left( 2^i-1  \right) = \\
	& = 2^{\frac{m(m-1)}{2}} \cdot \prod_{i=2}^{m} \left( 2^i-1  \right).
	\end{align*}
	It is worth noting that we rewrote the size of $\GL(m)$ as the product of the number of lower-triangular matrices $|\mathcal{L}|$ and the product of the first $n$ powers of two, diminished by one. 
	This property can be used to simplify the calculation of $\BLTA(S)$.  
	In fact, each block of the $\BLTA$ structure forms an independent $\GA(s_i)$ space, having size $|\GA(s_i)|$. 
	All the entries above the block diagonal are set to zero, so they are not taken into account in the size calculation, while the entries below the block diagonal are free, and can take any binary value. 
	Then, the size of $\BLTA(S)$ can be calculated as $|\LTA(n)| = 2^{\frac{n(n+1)}{2}}$ multiplied by the product of the first $s_i$ powers of two, diminished by one, for each size block $s_i$, which concludes the proof. 
	\end{proof}
\fi
\end{lemma}
\begin{lemma}\label{cor:size_EC_SC}
A polar code of length $N=2^n$ with automorphism group $\BLTA(S)$, $S=(s_1>1,\dots,s_t)$ has 
\begin{equation}\label{eq:nb_ec_sc}
|\EC_{SC}|=\frac{|\BLTA(S)|}{|\BLTA(2,1,\dots,1)|} = \frac{1}{3}\prod_{i=1}^t \left( \prod_{j=2}^{s_i} \left( 2^j-1 \right) \right)
\end{equation}
equivalence classes under aSC decoding.
\ifdefined\PROOFS
	\begin{proof}
	This follows from the application of Theorem~\ref{theo:EC} and Lemma~\ref{lem:sizeBLTA}.
	\end{proof}
\fi
\end{lemma}
%

\subsection{Generation of equivalence class representatives}\label{subsec:gen_EC}
A \emph{representative} $\pi$ of an EC is an element of the class; an EC can be represented through its representative as $[\pi]$. 
In this section, we describe how to find a representative for each EC under aSC decoding.
To begin with, we define the PUL decomposition of an affine transformation: the transformations $P \in \PL$, $U \in \UTL$ and $(L,b_0) \in \LTA$ are called the PUL decomposition of $(A,b)$ in \eqref{eq:GA} if
\begin{equation*}
	A v + b = P \cdot U \cdot (L v + b_0 ) ,
\end{equation*}
i.e., if $A = P U L$ and $b = P U b_0$.
As remarked in \cite{geiselhart2020automorphism}, the PUL decomposition of $(A,b)$ corresponds to the concatenation of the permutations as
\begin{equation}\label{eq:PUL_perm}
	\pi_{(A,b)} = \pi_{L,b_0} \circ \pi_U \circ \pi_P  ,
\end{equation}
applied from right to left. 
Since component $(L,b_0) \in \LTA$ is absorbed by SC, we can focus on the other components $P$ and $U$. 
If we call $\mathcal{A}_{\mathcal{P}}$ and $\mathcal{A}_{\mathcal{U}}$ the subsets of $\mathcal{A}$ containing only UTL transformations and permutation matrices, we can always find an EC representative composing elements from these two sets. 
Figure~\ref{fig:BLTA_EC} shows an example of pattern $A=PU$ with $P\in\mathcal{A}_{\mathcal{P}}$ and $ U\in\mathcal{A}_{\mathcal{U}}$.
\begin{theorem}\label{theo:PU_BLT}
Each EC contains at least one automorphism $P \cdot U$ with $P\in\mathcal{A}_{\mathcal{P}}$ and $U\in\mathcal{A}_{\mathcal{U}}$.
\begin{proof}
Given a representative $A$ for an equivalence class, it is possible to decompose this matrix as $A = P U L$. 
Since $L \in \LTA \leq [\pid]$, then $L^{-1} \in [\pid]$ and hence $A \cdot L^{-1} = PU$ belongs to the same EC as $A$, i.e., $[PU] = [A]$.
\end{proof}
\end{theorem}

\begin{figure}[t]
\centering
\begin{subfigure}{0.45\columnwidth}
  \centering
  \includegraphics[width=0.55\columnwidth]{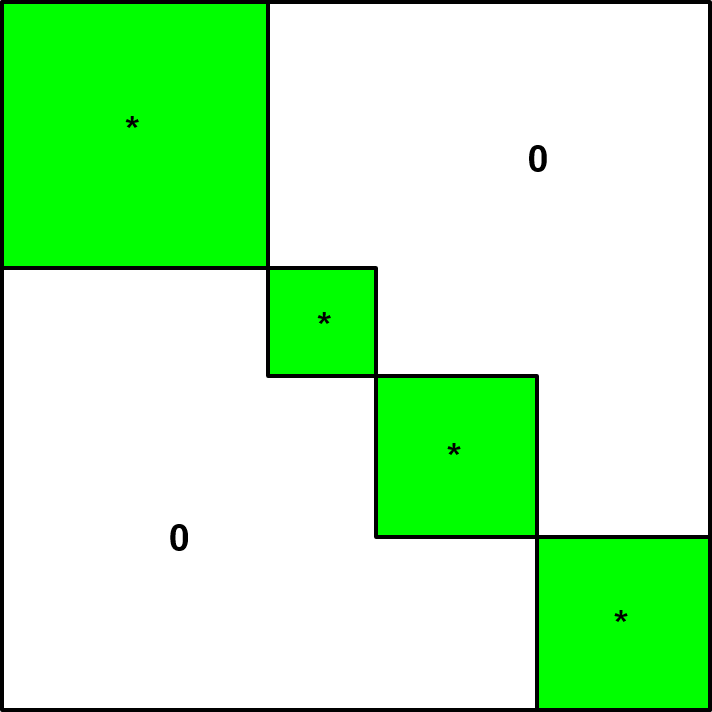}
  \caption{}
\label{fig:BLTA_EC}
\end{subfigure}
\begin{subfigure}{0.45\columnwidth}
  \centering
  \includegraphics[width=0.95\columnwidth]{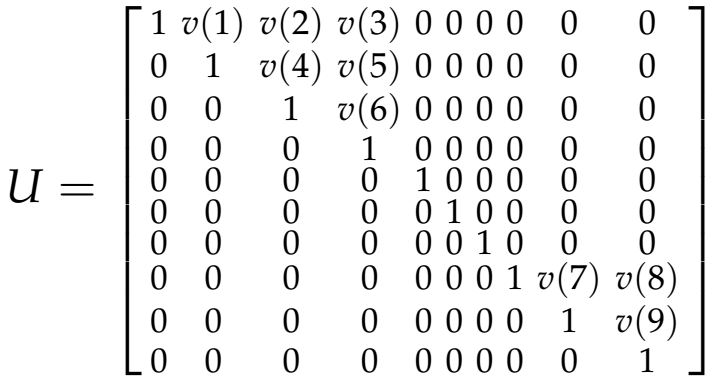}
  \caption{}
  \label{fig:BLT_with_u}
\end{subfigure}%
\caption{(a) BLT fitting every EC (b) bit order of $v$ in $U\in\mathcal{A}_{\mathcal{U}}$ with $S=(4,1,1,1,3)$.}
\end{figure}

\subsection{EC selection for AE decoders}
%
As discussed before, two EC may produce the same candidate codeword for particular values of $y$. 
Hence a method to select which EC to include in an AE decoder is of capital importance, especially for small values of $M$. 
Here we propose an heuristic approach for this selection based on the concept of distance among $\UTL$ and permutation matrices. 

For a block structure $S=(s_1,\dots,s_t)$, the number of automorphisms generated by a UTL transformation or from the permutation group are:
\begin{equation}\label{eq:aut_U_P}
	|\mathcal{A}_{\mathcal{U}}| = \prod_{i=1}^{t}2^{\frac{s_i(s_i-1)}{2}} \quad , \quad |\mathcal{A}_{\mathcal{P}}| = \prod_{i=1}^{t}(s_i!)~.
\end{equation}
We describe a permutation $P\in\mathcal{A}_{\mathcal{P}}$ with the vector $p$ of size $n$ where $p(i)=j \Leftrightarrow P(j,i)=1$. 
Similarly, every possible  $U\in\mathcal{A}_{\mathcal{U}}$ is described by a binary vector $v$ of size $m$ where $m=\text{log}_2(|\Aut(\mathcal{U}|)$; 
the bit order of $v$ follows the natural order of rows and then columns as shown in Figure~\ref{fig:BLT_with_u}.
Next, we define the Hamming distance (HD) between two vectors $a,a'$ of length $l$ as:
\begin{equation} \label{eq:HD_u_p}	
 HD\left(a, a' \right) = l-\sum_{i=1}^{l} \delta_{a(i),a'(i)}
\end{equation}
where $\delta_{a(i),a'(i)}$ is the Kronecker delta function being 1 if $a(i) = a'(i)$ and 0 otherwise. 

Our heuristic selection is based on a pair of Hamming distances thresholds $D=(d_{U},d_{P})$, representing the minimum distances between two EC representatives in the $M$ automorphisms selected for AE decoding. 
Every new EC representative is constructed by randomly generating two vectors $v$ and $p$, representing an $\UTL$ and a permutation matrices respectively, and checking if the distances between the generated vectors and the vectors already in the EC representative lists are above the thresholds. 
Moreover, a further check is required to assure that the EC representative does not belong to an already calculated EC. 
This check is done by multiplying the calculated $UP$ matrix and the $UP$ matrices of the previously calculated ECs as stated in Lemma~\ref{lem:checkproductA}. 
\begin{lemma}\label{lem:checkproductA}
Given $\pi_1,\pi_2 \in \mathcal{A}$ having transformation matrices $A_1,A_2$, then $\pi_1\in[\pi_2]$ if and only if $A_1 \cdot A_2^{-1} \in\BLTA(2,1,\ldots,1)$.  
\ifdefined\PROOFS
 	\begin{proof}
 	This follows dirctly from Lemma~\ref{lem:inv} and the definitions of EC and $[\pid]$. 
 	\end{proof}
\fi
\end{lemma}
If either one of the two checks fails, the vectors are discarded. 
The process is repeated until the desired number of automorphisms $M$ is reached.

\section{Automorphism-friendly design of polar codes}\label{sec:auto_design}
In this section, we propose a new polar code design conceived to match a desired affine automorphism group. 
This method is an evolution of what we proposed in \cite{PC_UTL_design}, and permits to overcome its main bottlenecks given by the need of human inspection and the lack of the decreasing monomial property of the resulting code. 
This new design can be completely automatized, and the decreasing monomial property is kept. 
The authors in \cite{geiselhart2021automorphismPC} propose to generate polar codes with a good affine automorphism group by starting from a minimal information set $\mathcal{I}_{min}$; however, this method does not permit to select the code parameters $N$ and $K$ in advance. 

The proposed design method requires 4 inputs, namely the code length $N$, the desired dimension $K$, the desired block structure $S$, and a starting design SNR $SNR_{min}$. 
The algorithm successfully stops when an $(N,K)$ polar code with the desired block structure $S$ is retrieved, while it fails if such a code cannot be designed. 
Its pseudo-code is shown in Algorithm~\ref{alg:design}. 
The proposed method uses an auxiliary matrix $A_\mathcal{G}$ of size $n \times n$ representing the monomials needed to be added to $\mathcal{G}$ in order to free a certain position $(i,j)$ in the affine transformation matrix $A$. 
the method to generate such a matrix from a given $\mathcal{I}$ is described in \cite{PC_UTL_design}. 
Algorithm~\ref{alg:design} runs two encapsulated loops: the external one is based on the generation of a reliability sequence $\mathcal{R}$ given the design SNR, while the internal one increases at every step the number of virtual channels to be inserted in $\mathcal{I}$ on the basis of automorphism properties instead of reliability, namely by picking monomials listed in entries of $A_\mathcal{G}$. 

\begin{algorithm}[t!]
\newcommand\mycommfont[1]{\footnotesize\ttfamily\textcolor{blue}{#1}}
\SetCommentSty{mycommfont}
		\SetKwInOut{Input}{input}%
		\SetKwInOut{Output}{output}%
		\SetKw{Return}{return}
		\SetAlgoLined
		\Input{Length $N$, dimension $K$, block structure $S=(s_1,\dots, s_t)$, design SNR $SNR_{min}$}
		\Output{Generating monomial set $\mathcal{G}$}
		\For{$SNR =SNR_{min}:\Delta_{SNR}:SNR_{max}$}{
		$\mathcal{R}  \leftarrow \text{DE/GA}\left(N,SNR\right)$\tcp*{Reliability list}	\label{alg:rel_sequence}
		$K_s\leftarrow K-1$\tcp*{\scriptsize Pre-design dimension}		
		\While{1}{
		$\mathcal{G}\leftarrow \mathcal{R}(1:K_s)$\;
		$d_{max} \leftarrow \text{MaxDegree}(\mathcal{G})$\;
		\If{$\sum_{k=0}^{d_{max}}{n\choose k}<K$}{
			break\;\label{alg:ET_alg}
		}
		\For{$i=1:1:t$}{\tcp{Free block by block}
			$col\leftarrow \sum_{b={1}}^{i-1}s_b +1$ \tcp*{$1^{st}$ column of $i^{th}$ block}
			\For{$C=col:1:col+s_i-1$}{
				$\mathcal{G}_a \leftarrow []$\tcp*{New monomials to add}
				\For{$dim = 1:1:s_i$}{
                        $\mathcal{G}_a\leftarrow \mathcal{G}_a \cup A_{\mathcal{G}}(col+dim-1,C)$\;
                 }        
                 $\mathcal{G} \leftarrow \mathcal{G}\cup \text{unique}(\mathcal{G}_a)$\;
                 $A_{\mathcal{G}} \leftarrow \text{ComputeAm}(\mathcal{G})$ \tcp*{$C^{th}$ column of $i^{th}$ block is free}\label{alg:check_column}
			}
			
			\If{$|\mathcal{G}|>K$}{
					break\;
			}	
		}
		\If{$|\mathcal{G}|=K$}{
			\Return{$\mathcal{G}$}
		}
		\Else{
					$K_s\leftarrow K_s-1$\;
		}
		}
		}
		\Return{Design failure\tcp*{code not achievable}}
	\caption{Proposed design}\label{alg:design}
\end{algorithm}

To begin with, a reliability sequence $\mathcal{R}$ for virtual channels is generated on the basis of $SNR_{min}$ using e.g. the DE/GA method \cite{frozenset}. 
Next, the $K_s=K-1$ most reliable monomials are used to create the initial monomial set $\mathcal{G}$ of the code. 
Then, the auxiliary matrix $A_\mathcal{G}$ is generated as described in \cite{PC_UTL_design}, where $A_\mathcal{G}(i,j)$ stores the list of monomials to be added to free position $(i,j)$. 
Then, our method starts adding monomials listed in $A_\mathcal{G}$ to match the desired first block dimension $s_1$. 
This addition is performed column by column and updating $A_\mathcal{G}$ after each unlocked column (Line~\ref{alg:check_column}).  
After the inclusion of all the monomials to free the first block, a check on the size of $\mathcal{G}$ is performed: if $|\mathcal{G}|\leq K$, the algorithm keeps $\mathcal{G}$ and proceeds to the next block; otherwise the procedure is restarted after decrementing $K_s$ by one. 
For the $i^{th}$ block, the procedure is repeated, namely the auxiliary matrix $A_\mathcal{G}$ is calculated on the basis of the monomial set $\mathcal{G}$ calculated for the previous block, all the monomials corresponding to $i^{th}$ block are added to $\mathcal{G}$ column by column and its size is checked: if $|\mathcal{G}|\leq K$, the algorithm proceeds to next block, otherwise the procedure restarts from the first block with $K_s=K_s-1$. 
The algorithm ends successfully if $|\mathcal{G}|=K$ at the check of the last block of size $s_t$, and the information set corresponding to the monomial set $\mathcal{G}$ is provided as output by the algorithm. 
If the design is not successful, a new attempt is started in the outer loop, now with the reliability sequence for an increased design SNR.
If the design SNR exceeds a given threshold $SNR_{max}$, the algorithm terminates with failure. 
 
The running time of the proposed algorithm can be reduced by early stopping the $K_s$ decrease by checking if the number of monomials to be added for the next blocks is too large. 
In fact, $A_\mathcal{G}$ can only provide monomials to pick with degrees that are lower or equal to the maximum degree already located in the generating monomial set. 
Since the number of monomials up to a degree $d$ is given by $\sum_{k=0}^{d}{n\choose k}$, if we call $d_{max}$ the largest degree present in $\mathcal{G}$, then $A_\mathcal{G}$ can provide enough monomials to reach the desired dimension $K$ if and only if $\sum_{k=0}^{d_{max}}{n\choose k}<K$. 
This check is performed in Line~\ref{alg:ET_alg}.

\section{Simulation Results}\label{sec:num}
In this section we present simulation results of polar codes transmitted with BPSK modulation over the AWGN channel. 
We show polar codes designed with Algorithm~\ref{alg:design}, in blue in the figures, under AE-$M$-decoding, where $M$ represents the number of parallels polar decoders in the AE, and their ML bound (determined by SCL decoding with list size $L=512$), comparing the results with state-of-the-art polar codes designed for AE-SC decoding and 5G polar codes \cite{5GHuawei} under CRC-aided SCL decoding with list size $L$.
\begin{figure}[t]
	\includegraphics[width=0.995\columnwidth]{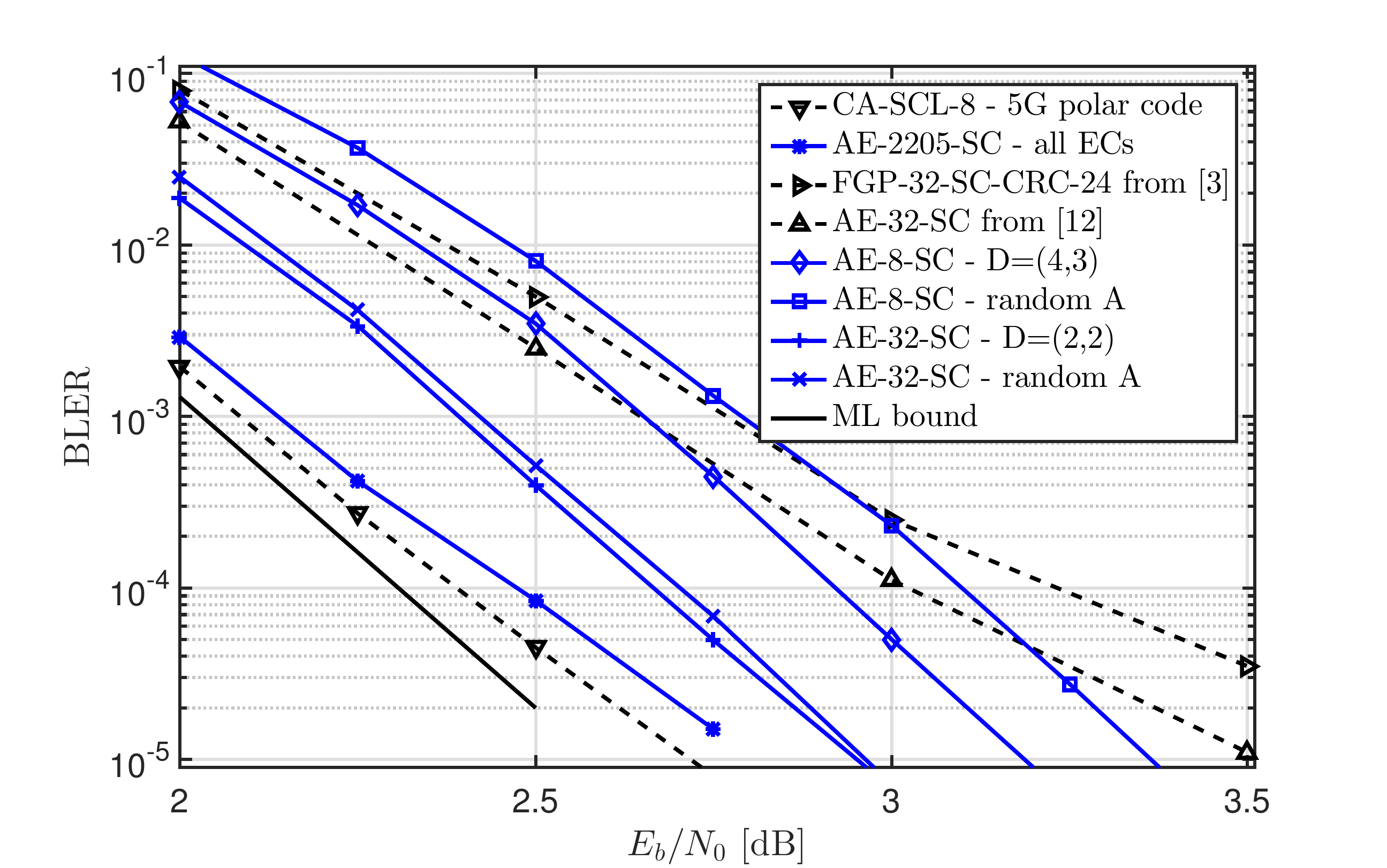}
	\caption{AE performance of $(1024,512)$ with $S=(4,1,1,1,3)$.}
	\label{fig:1024_512}
\end{figure}

Figure~\ref{fig:1024_512} compares the performance of $(1024,512)$ polar codes designed with Algorithm~\ref{alg:design} for $S=(4,1,1,1,3)$ to other permutation-based designs provided in \cite{PC_UTL_design,PermGross}, that may not be compliant with the UPO framework. 
Our proposal outperforms the other designs for AE-decoding, reducing the gap to the 5G polar codes decoded with CA-SCL. 
Moreover, the use of one automorphism from each of the $|EC|=2205$ equivalence classes permits to approach the ML bound for this code. 
Finally, the proposed method for EC selection works well for small values of $M$, while its contribution is reduced for larger values of $M$.

\begin{figure}[t!]
	\includegraphics[width=0.995\columnwidth]{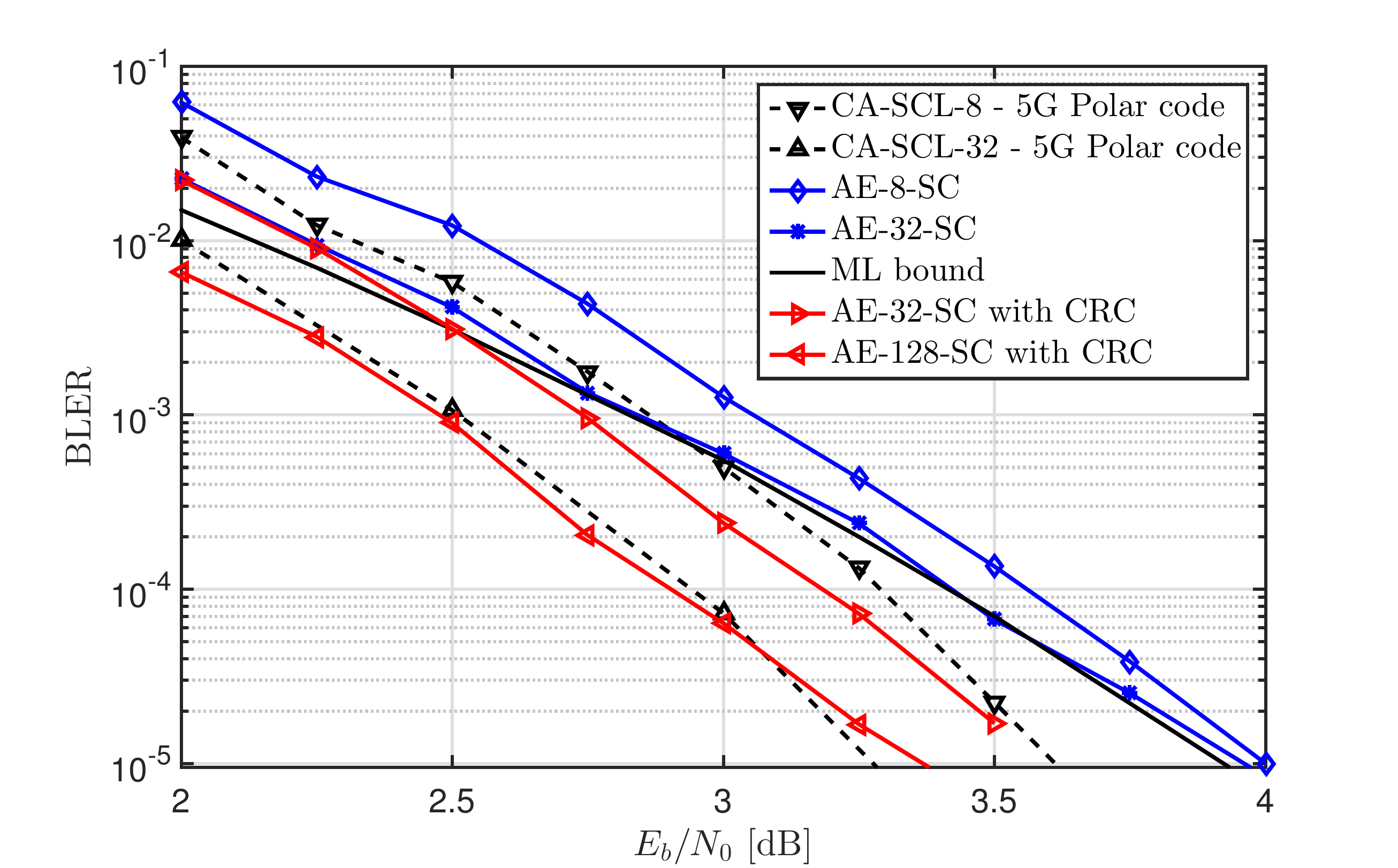}
	\caption{AE performance of $(256,128)$ code with $S=(3,5)$.}
	\label{fig:256_128}
\end{figure}
Figure~\ref{fig:256_128} shows the performance evaluation of $(256,128)$ polar codes with $S=(3,5)$; this code was studied in \cite{geiselhart2021automorphismPC,sym1} and our design allows to retrieve it. 
AE-32-SC with $D=(4,3)$ permits to reach ML performance, however the results are still far away from 5G polar codes. 
A more accurate analysis of AE-SC decoding results show that correct codewords are sometimes generated but discarded due to the least-square metric. 
In this case, the introduction of a CRC of 6 bits, used in 5G\cite{5GHuawei}, permits to largely improve the performance. 
However, it is not clear when the CRC is useful: in fact, its introduction for the code presented in Figure~\ref{fig:1024_512} does not improve the AE-SC decoding performance. 
We conjecture this performance gain to be connected to the number of equivalence classes, which in this case and for the code depicted in Figure~\ref{fig:256_128_53} is quite large, namely $|EC|=68355$. 
 
Finally, Figure~\ref{fig:256_128_53} shows the performance of another $(256,128)$ polar code with $S=(5,3)$, known for having poor performance under AE-SC decoding\cite{geiselhart2021automorphismPC}. 
This result is confirmed by our simulations. 
The approach proposed in \cite{Perm_Russian_polar_subcode}, introducing dynamic frozen bits and using permutations that are not code automorphisms, permits to improve the AE-SC decoding performance, approaching the ML bound. 
Here we propose another approach, based on the use of soft decoders; AE-SCAN with $5$ and AE-BP with $100$ iterations reaches ML performance for $M=32$ automorphisms selected based on the Hamming distance constraint with $D=(4,3)$, however presuming $[\pid] = \{\pid\}$. 
The introduction of the CRC of 6 bits permits to improve the performance at small BLER, beating the ML bound. 
Since both codes share the same number of equivalence classes, we conjecture as well that the symmetry of the code influences the performance gain provided by a CRC.

\section{Conclusions}\label{sec:conclusions}
In this paper, we proposed a polar code design allowing for a desired affine automorphism structure. 
We expand the SC-absorption group to $\BLTA(2,1,\dots,1)$ for most of the polar codes, while calculating the maximum number of non-redundant automorphisms under AE-SC decoding. 
We introduce the notion of equivalence relation under AE-based decoding, which represents a powerful tool for the study of AE decoding of monomial codes. 
The classification of automorphisms for other decoders remains an open problem; we conjecture that AE-SC and AE-SCL have the same equivalence class structure, while $[\pid] = \{\pid\}$ for AE-BP and AE-SCAN. 
 \begin{figure}[t!]
	\includegraphics[width=0.995\columnwidth]{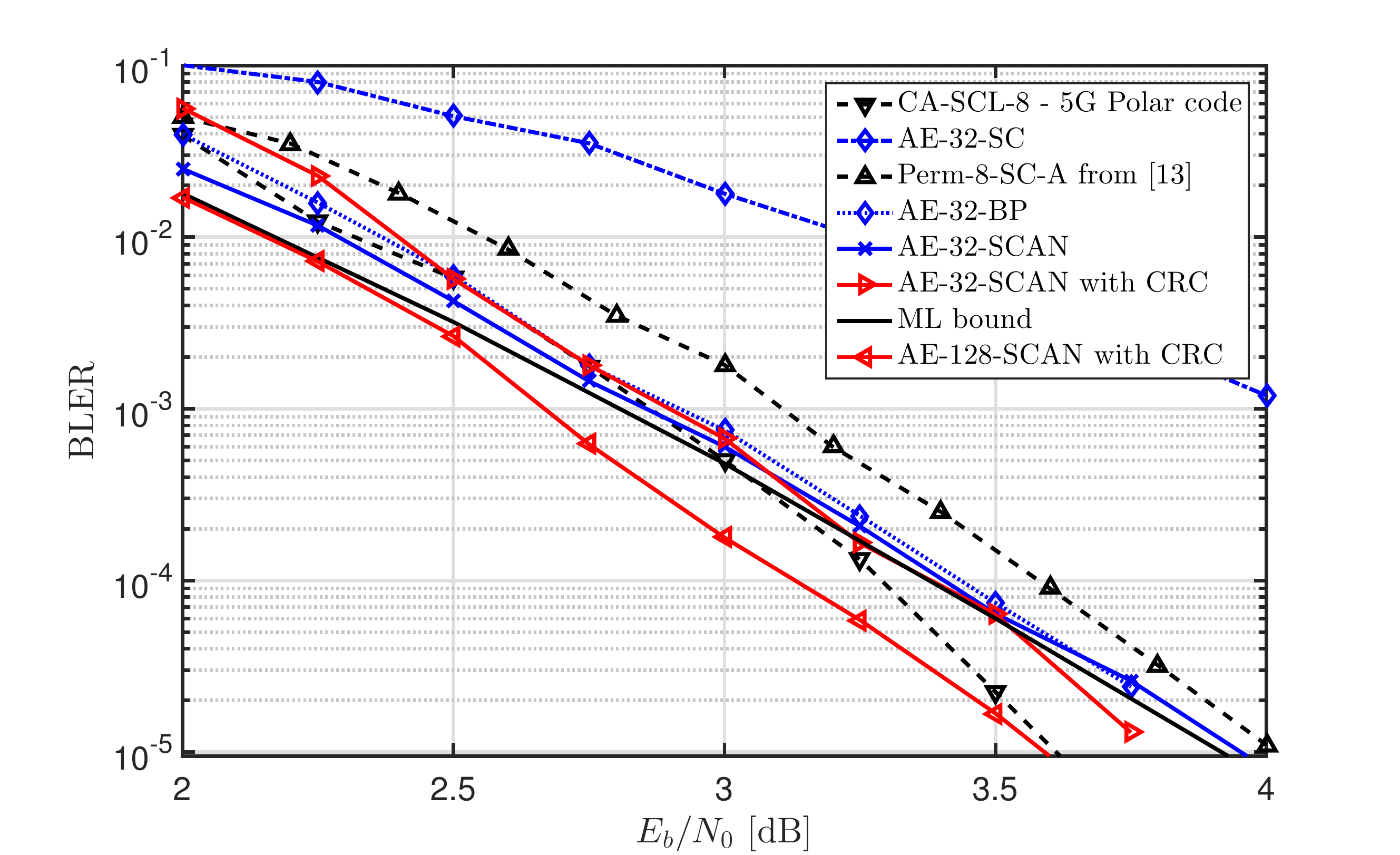}
	\caption{AE performance of $(256,128)$ code with $S=(5,3)$.}
	\label{fig:256_128_53}
\end{figure}

\bibliographystyle{IEEEbib}
\bibliography{references}

\end{document}